\documentclass[a4paper,11pt]{amsart}

\usepackage[english]{babel}
\usepackage[T1]{fontenc}
\usepackage[utf8x]{inputenc}
\usepackage{xspace}
\usepackage{hyperref}
\usepackage{subfigure}
\usepackage{color}
\usepackage{mathrsfs}
\usepackage{framed}

\usepackage{mathrsfs}
\usepackage{amsfonts}
\usepackage{amssymb}
\usepackage{amsmath}
\usepackage{array,multirow}
\usepackage{mathbbol}

\usepackage{color}
\usepackage{graphicx}
\usepackage{graphics}

\usepackage{mathabx}

\def\cO{\mathcal{O}}

\def\cC{\mathcal{C}}

\def\cH{\mathcal{H}}
\def\cE{\mathcal{E}}
\def\cP{\mathcal{P}}

\def\cT{\mathcal{T}}
\def\cP{\mathcal{P}}
\def\cL{\mathcal{L}}

\def\cF{\mathcal{F}}

\def\dom{{\sc Dom-Enum}\xspace}

\def\np{{\sc NP}\xspace}
\def\extensionproblem{\textsf{\sc Extension Problem}\xspace}
\def\sat{{\sc Sat}\xspace}

\newcommand{\eg}{\emph{e.g.}\xspace}
\newcommand{\ie}{\emph{i.e.}\xspace}

\newcommand{\macro}[3]{\newcommand{#1}[#3]{#2}}
\macro{\size}{\|#1\|}{1}

\newtheorem{thm}{Theorem}
\newtheorem{lem}[thm]{Lemma}

\newtheorem{prop}[thm]{Proposition}

\newtheorem{property}[thm]{Property}
\newtheorem{fact}[thm]{Fact}

\newtheorem{theorem}[thm]{Theorem}
\newtheorem{lemma}[thm]{Lemma}

\title[Poly Delay for \dom in Chordal Graphs]{A Polynomial Delay Algorithm for Enumerating Minimal Dominating Sets in Chordal Graphs}
\author[M.M. Kant{\'e} \and V. Limouzy\and A. Mary \and L. Nourine \and T. Uno]{Mamadou Moustapha Kant{\'e} \and Vincent Limouzy \and Arnaud
  Mary \and Lhouari Nourine \and Takeaki Uno}

\address{Clermont-Universit{\'e}, Universit{\'e} Blaise Pascal, LIMOS,
  CNRS, France} 
\email{\{mamadou.kante,limouzy,mary,nourine\}@isima.fr}
\address{National Institute of Informatics, Japan}
\email{uno@nii.jp}
\thanks{M.M. Kant{\'e} and V. Limouzy are supported by the French Agency
  for Research under the DORSO project.}

\begin{document}

\begin{abstract} An output-polynomial algorithm for the listing of minimal dominating sets in graphs is a challenging open problem and is known to be equivalent to the well-known Transversal
  problem which asks for an output-polynomial algorithm for listing the set of minimal hitting sets in hypergraphs. We give a polynomial delay algorithm to list the set of minimal dominating
  sets in chordal graphs, an important and well-studied graph class where such an algorithm was open for a while. 
\end{abstract}

\maketitle

\section{Introduction} \label{sec:0} 

An \emph{enumeration algorithm} for a set $\mathscr{C}$ is an algorithm that lists the elements of $\mathscr{C}$ without repetitions.  A \emph{hypergraph} $\cH$ is a pair $(V,\cE)$ where $V$ is a
finite set and $\cE\subseteq 2^V$ is called the set of \emph{hyper-edges}. Hypergraphs generalize graphs where each hyper-edge has size at most $2$. Given a hypergraph $\cH:=(V,\cE)$ and $\cC\subseteq
2^V$ an \emph{output-polynomial algorithm} for $\cC$ is an enumeration algorithm for $\cC$ whose running time is bounded by a polynomial depending on the sum of the sizes of $\cH$ and $\cC$. The
enumeration of minimal or maximal subsets of vertices satisfying some property in a (hyper)graph is a central area in graph algorithms and for several properties output-polynomial algorithms have
been proposed \eg \cite{AvisF96,BorosEG06,EGM03,FominHKPV14,lawler80,SchwikowskiS02,Tarjan73}, while for others it was proved that no output-polynomial algorithm exists unless P=NP
\cite{KhachiyanBBEG08b,KhachiyanBBEGM08a,KhachiyanBEG08,lawler80,Strozecki10}.

One of the central problem in the area of enumeration algorithm is the existence of an output-polynomial algorithm for the set of \emph{minimal transversals} in hypergraphs, and is known as the
\emph{Transversal problem} or \emph{Hypergraph dualization}.  A \emph{minimal transversal} (or \emph{hitting set}) in a hypergraph $(V,\cE)$ is an inclusion-wise minimal subset $T$ of $V$ that intersects with every hyper-edge in $\cE$. The
Transversal problem has several applications in artificial intelligence \cite{EG95,EGM03}, game theory \cite{Gurvich73,ramamurthy90}, databases \cite{Agrawal96,BG00,BorosGKM00,ullman89}, integer
linear programming \cite{BG00,BorosGKM00}, to cite few. Despite the interest in the Transversal problem the best known algorithm is the quasi-polynomial time
algorithm by Fredman and Khachiyan which runs in time $O(N^{\log(N)})$ where $N$ is the cumulated size of the given hypergraph and its set of minimal transversals. However, there exist several classes
of hypergraphs where an output-polynomial algorithm is known (see for instance \cite{EG95,EGM03,KanteLMN12} for some examples). Moreover, several particular subsets of vertices in graphs are special
cases of transversals in hypergraphs and for some of them an output-polynomial algorithm is known, \eg maximal independent sets, minimal vertex-covers, maximal (perfect) matchings,
spanning trees, etc.

In this paper we are interested in the particular case of the Transversal problem, namely the enumeration of \emph{minimal dominating sets} in graphs (\dom problem). A \emph{minimal dominating set} in
a graph is an inclusion-wise subset $D$ of the vertex set such that every vertex is either in $D$ or has a neighbor in $D$. In other words $D$ is a minimal dominating set of $G$ if it is a minimal
transversal of the \emph{closed neighborhoods} of $G$. The \emph{closed neighborhood $N[x]$ of a vertex $x$} is the set containing $x$ and its neighbors. Since in important graph classes an
output-polynomial algorithm for the \dom problem is a direct consequence of already tractable cases for the Transversal problem, \eg minor-closed classes of graphs, graphs of bounded degree, it is
natural to ask whether an output-polynomial algorithm exists for the \dom problem. However, it is proved in \cite{KanteLMN12} that there exists an output-polynomial algorithm for the \dom problem if
and only if there exists one for the Transversal problem, and this remains true even if we restrict the \dom problem to the co-bipartite graphs.  This is surprising, but has the advantage of bringing
tools from graph structural theory to this difficult problem and is particularly true for the \dom problem since in several graph classes output-polynomial algorithms were obtained using the structure
of the graphs: graphs of bounded clique-width \cite{Courcelle09}, split graphs \cite{KanteLMN11,KanteLMN12}, interval and permutation graphs \cite{KanteLMNU13}, line graphs \cite{KLMN12,KanteLMNU14},
etc.

Since the \dom problem in co-bipartite graphs is as difficult as the Transversal problem and co-bipartite graphs are a subclass of weakly chordal graphs, \ie graphs with no cycles of length greater
than or equal to $5$, one can ask whether by restricting ourselves to graphs without cycles of length $4$, which are exactly \emph{chordal} graphs \cite{Dirac61}, one cannot expect an
output-polynomial algorithm. In fact for several subclasses of chordal graphs an output-polynomial algorithm is already known, \eg undirected path graphs \cite{KLMN12}, split graphs, chordal
$P_6$-free graphs \cite{KanteLMN11,KanteLMN12}. Furthermore, chordal graphs have a nice structure, namely the well-known \emph{clique tree} which has been used to solve several algorithmic questions
in chordal graphs.  We prove the following.

\begin{thm}\label{thm:main} There exists a polynomial delay algorithm for the \dom problem in chordal graphs which uses polynomial space.
\end{thm}

An output-polynomial algorithm is \emph{polynomial delay} if the delay between two outputs is bounded by a polynomial in the size of the input (we also require the times before the first output and
after the last output to be bounded by polynomials on the input size). Among output-polynomial algorithms polynomial-delay algorithms are the most desirable since they allow to treat the solutions as
they appear and we do not need to wait a long time between two outputs.  Notice that there exist problems where an output-polynomial algorithm is known and no polynomial delay algorithm exists unless
P=NP \cite{Strozecki10}.

It is well-known that every $n$-vertex chordal graph $G$ admits a linear ordering $x_1,\ldots,x_n$ of its vertex such that for every $1\leq i \leq n$ the vertex neighborhood of $x_i$ in
$G[\{x_i,\ldots,x_n\}]$ is a clique. For the enumeration of minimal dominating sets in chordal graphs the simplest strategy consists in following this ordering as follows. Since $N[x_1]$ is a
clique, any minimal dominating set of $G$ either contains $x_1$ or does not contain $x_1$ but contains at least one of its neighbors. Therefore any minimal dominating set of $G$ is either of the form
$D\cup \{x_1\}$ where $D$ is a minimal dominating set of $G\setminus N[x_1]$, or is a minimal dominating set of $G\setminus \{x\}$ that intersects the neighborhood of $x_1$. Unfortunately if the
first case is just a recursive call, it is an exercice to see that the Transversal problem reduces to the enumeration of minimal dominating sets of the second kind. Indeed, such a bottom-up strategy
is hopeless since we will face the problem of identifying which sets in sub-trees are extendable to minimal dominating sets. An idea would be then to follow the clique tree in a top-down way, but as
we will see if we do not take care we will come across an \np\!\!-complete problem. Our strategy will nevertheless follow the clique tree in a top-down way, but not in the usual way combining a kind
of breadth-first search and depth-first search of the tree.  We postpone the details of the strategy to forthcoming sections.

\medskip
\paragraph*{\bf Summary.} Definitions and preliminary results are given in Section \ref{sec:preliminary}. The strategy and some faced difficulties are presented in Section \ref{sec:difficulty}.  The algorithm and some
necessary technical lemmas are given in Sections \ref{sec:kextension} and \ref{sec:algo}. We conclude with some open questions. 

\section{Preliminaries} \label{sec:preliminary}

\subsection{General Definitions and Notations} 

We refer to \cite{Diestel2005} for our graph terminology. We deal only with finite simple loopless undirected graphs. The vertex set of a graph $G$ is denoted by $V_G$ and its edge set by $E_G$. An
edge between two vertices $x$ and $y$ is denoted by $xy$ ($yx$ respectively).  Let $G$ be a graph. The subgraph of $G$ induced by $X\subseteq V_G$, denoted by $G[X]$ is the graph $(X,(X\times X)\cap
E_G)$. The size of a graph $G$, denoted by $\size{G}$, is $|V_G|+|E_G|$, and the \emph{size} of any $\cC\subseteq 2^{V_G}$, denoted by $\size{\cC}$, is defined as $\sum_{C\in \cC} |C|$. For a vertex
$x$ of $G$ we denote by $N(x)$ the set of neighbors of $x$, \ie the set $\{y\in V_G\mid xy\in E_G\}$, and we let $N[x]$, the \emph{closed neighborhood of $x$}, be $N(x)\cup \{x\}$. For
$S\subseteq V_G$, let $N[S]$ denote $\bigcup_{x\in S}N[x]$. (We will remove the subscript when the graph is clear from the context and this will be the case for all sub or superscripts in the
paper.)  We say that a vertex $x$ is \emph{dominated} by a vertex $y$ if $x\in N[y]$.  A \emph{dominating set} of $G$ is a subset $D$ of $V_G$ such that every vertex of $G$ is dominated by a vertex
in $D$.  A dominating set is \emph{minimal} if it includes no other dominating set. For $D\subseteq V_G$, a vertex $x$ is a \emph{private neighbor} of $y\in D$ if $N[x]\cap D=\{y\}$; the set of
private neighbors of a vertex $x\in D$ is denoted by $P(D,x)$. $D\subseteq V_G$ is an \emph{irredundant set} of $G$ if $P(D,x)\ne \emptyset$ for all $x\in D$. The following is easy to obtain.

\begin{fact}\label{fact:min-dom-irr} $D\subseteq V_G$ is a minimal dominating set of $G$ if and only if $D$ is a dominating set of $G$ and $D$ is an irredundant set. 
\end{fact}

A \emph{clique} of $G$ is a subset $C$ of $G$ that induces a complete graph, and a \emph{maximal clique} is a clique $C$ of $G$ such that $C\cup \{x\}$ is not a clique for all $x\in V_G\setminus
C$. We denote by $\cC_G$ the set of maximal cliques of $G$.

A \emph{tree} is an acyclic connected graph. Since we will talk at the same time about a graph and a tree representing it the vertices of trees will be called \emph{nodes}. A \emph{rooted tree} is a
tree with a distinguished node, called its \emph{root}, and let us denote by $\preceq_T$ the relation on a rooted tree $T$, where $u\preceq_T v$ if $v$ is on the unique path from the root to $u$; if
$u\preceq_T v$ then $v$ is called an \emph{ancestor} of $u$ and $u$ a \emph{descendant} of $v$. Two nodes $u$ and $v$ of a rooted tree $T$ are \emph{incomparable} if $u\not\!\preceq_T v$ and
$v\not\!\preceq_T u$. Given a node $u$ of a rooted $T$ the subtree of $T$ rooted at $u$ is the tree $T[\{v\in V_T\mid v\preceq_T u\}]$ which is rooted at $u$. A graph $G$ is called \emph{chordal} if
it does not contain chordless cycles of length greater than or equal to $4$.

Let $G$ be a graph and let $\cC$ be a subset of $2^{V_G}$. An \emph{output-polynomial} algorithm for $\cC$ is an algorithm that lists the elements of $\cC$ without repetitions in time
$\cO\left(p\left(\size{G}, \size{\cC}\right)\right)$ for some polynomial $p$.  We say that an algorithm enumerates $\cC$ with \emph{polynomial delay} if, after a pre-processing that runs in time
$\cO(p(\size{G}))$ for some polynomial $p$, the algorithm outputs the elements of $\cC$ without repetitions, the delay between two consecutive outputs being bounded by $\cO(q(\size{G}))$ for some
polynomial $q$ (we also require that the time between the last output and the termination of the algorithm is bounded by $\cO(q(\size{\cH}))$). It is worth noticing that an algorithm which enumerates
a subset $\cC$ of $2^{V_G}$ in polynomial delay outputs the set $\cC$ in time $\cO\left(p(\size{G}) + q(\size{G})\cdot |\cC| + \size{\cC}\right)$ where $p$ and $q$ are respectively the polynomials
bounding the pre-processing time and the delay between two consecutive outputs.  Notice that any polynomial delay algorithm is obviously an output-polynomial one, but not all output-polynomial
algorithms are polynomial delay \cite{Strozecki10}. We say that an output-polynomial algorithm uses polynomial space if there exists a polynomial $q$ such that the space used by the algorithm is
bounded by $q(\size{G})$.

\subsection{Clique Trees of Chordal Graphs} 

An \emph{intersection graph} is a graph in which each vertex corresponds to a set and two vertices are adjacent if and only if their corresponding sets intersect. The collection of sets in
correspondence with the vertices of an intersection graph is called an \emph{intersection model}.  Chordal graphs are exactly intersection graphs of subtrees in trees \cite{Gavril1974}. A chordal
graph $G$ admits at most $|V_G|$ maximal cliques. From \cite{Gavril1974} to every chordal graph $G$, one can associate a tree that we denote by $\cT_G$, called \emph{clique tree}, whose nodes are in
bijection with the maximal cliques of $G$ and such that for every vertex $x\in V_G$ the set $\cT_G(x):=\{u\in V(\cT_G)\mid$ the maximal clique of $G$ corresponding to $u$ contains $x\}$ is a subtree
of $\cT_G$. Moreover, $G$ is the intersection graph of $\{\cT_G(x)\mid x\in V_G\}$. Notice that there exist several clique trees for every chordal graph $G$, but we can compute one in linear time (see
for instance \cite{GalinierHP95}). 
Let us now consider some properties of clique trees. First of all, since the nodes
of a clique tree $\cT$ of $G$ are in bijection with the maximal cliques of $G$ each node of $\cT$ will be identified with the maximal clique with which it is in correspondence. In the rest of the
paper all trees are considered rooted.

Let $\cT_G$ be a clique tree of a chordal graph $G$ and let us denote its root by $C_r$. For each $C\in \cC_G$, let us denote by $Pa(C)$ its parent and let $f(C):=C\setminus Pa(C)$, \ie, the set of
vertices in $C$ that are not in any maximal clique $C'$ ancestor of $C$.  Notice that $\{f(C)\mid C\in \cT_G\}$ is a partition of $V_G$. For each vertex $x\in V_G$, we denote by $C(x)$ the maximal
clique $C$ satisfying $x\in f(C)$. Notice that $C(x)$ is uniquely defined since exactly one maximal clique $C$ satisfies $x\in f(C)$. For $C\in \cC_G$, the subtree rooted at $C$ is denoted by
$\cT_G(C)$, and the set of vertices $\bigcup_{C'\in \cT_G(C)} f(C')$ is denoted by $V(C)$.

\begin{property}\label{property:cliquetree} Any clique tree $\cT_G$ of a chordal graph $G$ satisfies the following. 
\begin{enumerate}
\item For each $C\in \cC_G$, and each $x\in V_G\setminus V(C)$ either $(\{x\}\times f(C)) \subseteq E_G$ or $(\{x\} \times f(C)) \cap E_G=\emptyset$.
\item For any two incomparable $C$ and $C'$ in $\cC_G$, we have $(f(C)\times f(C')) \cap E_G=\emptyset$.
\end{enumerate}
\end{property}

For $S\subseteq V_G$ let $\cC(S)$ denote the set $\{C(x)\mid x\in S\}$, $Up(S)$ the set of vertices $x$ in $V_G$ such that $C(x)$ is a proper ancestor of a clique $C\in \cC(S)$ and $Uncov(S)$ be the
vertex set $Up(S)\setminus N[S]$, \ie the set of vertices in $Up(S)$ not dominated by $S$.
For a vertex $x$, $Up(x)$ denotes $Up(\{ x\})$. A subset $A\subseteq V_G$ is an \emph{antichain} if~ (1) for any two vertices $x$ and $y$ in $A$ we have
$x\notin Up(y)$ and $y\notin Up(x)$, ~(2) for each vertex $z\in V_G\setminus Up(A)$, $A\cap (C(z)\cup Up(z)) \ne \emptyset$.  Intuitively, $A$ is an antichain if $\cC(A)$ is a maximal set of pairwise
incomparable maximal cliques. Given $S\subseteq V_G$, the \emph{top antichain} $A(S)$ is defined as the set of vertices of $S$ included in the upmost cliques in $\cC(S)$ that are not descendants of
any other in $\cC(S)$, i.e., $A(S) := \{ x\in S \mid C(x) \mbox{ is in } \max\limits_{\preceq_\cT}\{ \cC(S)\} \}$.

If $S\ne \emptyset$, let $\cL(S)$ be the set of maximal cliques $C$ satisfying~ (1) no descendant of $C$ is in $\cC(S)$, ~(2) some descendants of $Pa(C)$ is in $\cC(S)$.  In other words, $\cL(S)$ is
the set of upmost maximal cliques no descendant of which intersects with $\cC(S)$, \ie, $\cL(S):=\max\limits_{\preceq_\cT}\{C\in \cC_G\mid C\ \mbox{has no descendant in}\ \cC(S)\}$.  If $S=\emptyset$,
let $\cL(S)$ be $\{C_r\}$. We denote by $\cL'(S)$ the set $\max\limits_{\preceq_\cT}\{C'\in \cT(C)\mid C\in \cL(S)\ \mbox{and}\ C'\cap S=\emptyset\}$. 

We suppose that any clique tree $\cT$ is numbered by a pre-order of the visit of a depth-first search.  In this numbering, the numbers of the nodes in any subtree forms an interval of the numbers. It
is worth noticing that this ordering is a linear extension of the descendant-ancestor relation.  We say that a clique is smaller than another clique when its number is smaller than the other's.  We also
extend this numbering to the vertices of the corresponding graph so that the number of a vertex $x$ is smaller than that of a vertex $y$ if $C(x)$ is smaller than $C(y)$.  We also say that a vertex is
smaller than another vertex if its number is smaller than the other's.  For a vertex set $S$, $tail(S)$ denote the largest vertex in $S$.  A \emph{prefix} of a vertex set $S$ is its subset $S'$ such
that no vertex in $S\setminus S'$ is smaller than $tail(S')$.  A \emph{partial antichain} is a prefix of an antichain. We allow the $\emptyset$ to be a partial antichain.

Following this ordering of the vertices of a chordal graph $G$, a minimal dominating set $D$ is said to be \emph{greedily obtained} if we initially let $D:=V_G$ and recursively apply the following
rule: if $D$ is not minimal, find the smallest vertex $x$ in $D$ such that $D\setminus \{x\}$ is a dominating set and set $D:=D\setminus \{x\}$. Notice that given a graph $G$ there is one greedily
obtained minimal dominating set. 

\section{When Simplicity Means \np\!\!-Hardness}\label{sec:difficulty}


A typical way for the enumeration of combinatorial objects is the \emph{back tracking} technique.  We start from the emptyset, and in each iteration, we choose an element $x$, and partition the
problem into two subproblems: the enumeration of those including $x$, and the enumeration of those not including $x$, and recursively solve these enumeration problems.  If we can check the so called
\extensionproblem in polynomial time, then the algorithm is polynomial delay and uses only polynomial space.  The \extensionproblem is to answer the existence of an object including $S$ and that does
not intersect with $X$, where $S$ is the set (partial solution) that we have already chosen in the ancestor iterations, and that includes all elements we decided to put in the output solution, and $X$
is the set that we decided not to include in the output solution.

It is known that the \extensionproblem for minimal dominating set enumeration is \np\!\!-complete \cite{Mary13}, and one can even prove that it is still \np\!\!-complete in split graphs (Proposition
\ref{prop:np-split}), which are a proper subclass of chordal graphs.  However, split graphs have a good structure and in the paper \cite{KanteLMN11}, it is proved that if $S\cup X$ induces a clique
the \extensionproblem in split graphs can be solved in polynomial time and this combined with the structure of minimal dominating sets in split graphs lead to a polynomial delay algorithm for the \dom
problem in split graphs. Chordal graphs also have a good tree structure induced by clique trees.  Thus, by following this tree structure, the \extensionproblem seems to be solvable.  In precise, we
consider the case in which a path $\cP$, from the root, of the clique tree satisfies that both $V(C)\cap (S\cup X) \ne \emptyset$ and $V(C)\not\subseteq (S\cup X)$ holds only for cliques $C$ included
in $\cP$.  In other words, the condition is that for any clique $C\not\in \cP$ whose parent is in $\cP$, either $V(C)\cap (S\cup X)= \emptyset$ (totally not determined) or $V(C)\subseteq (S\cup X)$ (totally
determined) holds.  The solutions are partially determined on the path $\cP$, and thus the \extensionproblem seems to be polynomial.  However, Theorem \ref{thm:np-chordal} states that the problem is
actually \np\!\!-complete.

\begin{prop}\label{prop:np-split} The \extensionproblem is \np\!\!-complete in split graphs. \end{prop}

\begin{proof} It is proved in \cite{Mary13} that the following problem is \np\!\!-complete: Given $G$ and $A\subset V_G$ decide whether there exists a minimal dominating set of $G$ containing $A$. We
  reduce it to the \extensionproblem in split graphs. Let $G$ be a graph, and let $V_G':=\{x'\mid x\in V_G\}$ a disjoint copy of $V_G$. We let $Split(G)$ be the split graph with vertex set $V_G\cup
  V_G'$ where $V_G$ and $V_G'$ are respectively the clique and the independent set in $Split(G)$; now $xy'$ is an edge if $x\in N[y]$.  Now it is easy to check that asking whether there exists a
  minimal dominating set of $G$ that contains $A\subset V_G$ is equivalent to asking whether there exists a minimal dominating set of $Split(G)$ that contains $A$ and does not intersect with
  $V_G'\setminus A'$ where $A':=N_{Split(G)}[A]\cap V_G'$. \end{proof}

\begin{thm}\label{thm:np-chordal} The \extensionproblem is \np\!\!-complete in chordal graphs even if a path $\cP$, from the root, of the clique tree satisfies that any child $C$ of a clique in $\cP$ satisfies either $V(C)\cap (S\cup
  X)=\emptyset$ or $V(C)\subseteq (S\cup X)$.
\end{thm}

\begin{proof} We reduce \sat to our problem.  Let $\varphi$ be an instance of \sat with $x_1,\ldots,x_n$ the variables and $c_1,\ldots,c_m$ the clauses of $\varphi$.  We construct a chordal graph as
  follows.  The vertex set of the graph is
  \begin{align*}
    & \{ x_1,\dots,x_n, c_1,\ldots,c_m, p_1,\ldots,p_n,\bar{p}_1,\ldots,\bar{p}_n, l_1,\ldots,l_n\}\  \bigcup \\ & \{\bar{l}_1,\ldots,\bar{l}_n, y_1,\ldots,y_n,z_1,\ldots,z_n, q_1,\ldots,q_n,
    \bar{q}_1,\ldots,\bar{q}_n\}, 
  \end{align*}
  where $l_i$ and $\bar{l}_i$ are literals representing respectively $x_i$ and $\widebar{x_i}$ (notice that if one literal does not appear, the corresponding vertex is not created).  Since with every
  clique tree one can associate a unique chordal graph, we will construct the clique tree of the chordal graph.  For each $1\leq i \leq n$, we let $C(l_i)$ and $C(\bar{l}_i)$ be the set of clauses
  containing the literal $l_i$ and $\bar{l}_i$ respectively.  We let its root be $C_r := \{ c_1,\ldots,c_m, p_1,\ldots,p_n,\bar{p}_1,\ldots,\bar{p}_n \}$. The other maximal cliques are defined as
  follows. For each $1\leq i \leq n$, we let $C_{x_i} = \{ x_i, p_i, \bar{p}_i\}$, $C_{y_i} = \{ y_i, x_i\}$, $C_{z_i} = \{ y_i, z_i\}$, $C_{q_i} = \{q_i, l_i\}$, $C_{\bar{q}_i} = \{\bar{q}_i,
  \bar{l}_i\}$, $C_{l_i} = \{ l_i, p_i\}\cup C(l_i)$, and $C_{\bar{l}_i} = \{ \bar{l}_i, \bar{p}_i\}\cup C(\bar{l}_i)$ with the following parent-child relation: $C_{x_i}$, $C_{l_i}$ and
  $C_{\bar{l}_i}$ are the children of $C_r$, $C_{y_i}$ is the only child of $C_{x_i}$ and $C_{z_i}$ is the only child of $C_{y_i}$, $C_{q_i}$ and $C_{\bar{q}_i}$ are the only children of $C_{l_i}$ and
  $C_{\bar{l}_i}$ respectively. It is easy to check that the constructed tree is indeed a clique tree. See Figure \ref{fig:np-chordal} for an illustration.

  We set $S := \{x_1,\ldots,x_n, y_1,\ldots,y_n\}$ and $X := \{ z_1,\ldots, z_n, p_1,\ldots,p_n,\bar{p}_1,\ldots,\bar{p}_n\} \cup$\\ $\{c_1,\ldots,c_m\}$ and $\cP:=\{C_r\}$.  For each $1\leq i \leq
  n$, we have by construction $V(C_{x_i})\subseteq S\cup X$, and $(V(C_{l_i}) \cup V(C_{\bar{l}_i})) \cap (S\cup X) = \emptyset$. Therefore, for any maximal clique $C$ child of $C_r$, either $V(C)\cap
  (S\cup X) = \emptyset$, or $V(C)\subseteq (S\cup X)$ holds, thus the condition of the statement holds.

  One can easily check that any satisfiable assignment of $\varphi$ leads to a minimal dominating set containing $S$ and that does not intersect $X$. Let us prove the converse direction. We observe when we
  choose both $l_i$ and $\bar{l}_i$ in the dominating set, $x_i$ loses its private neighbors.  Thus, any minimal dominating set can include at most one of them.  On the other hand, exactly one of
  $l_i$ and $q_i$ (resp., $\bar{l}_i$ and $\bar{q}_i$) must be included in any minimal dominating set, so that it dominates $l_i$ and $q_i$ (resp., $\bar{l}_i$ and $\bar{q}_i$), and both must be
  private neighbors of the chosen one.  Moreover, to dominate each clause $c_j$, at least one literal of $c_j$ has to be included in any minimal dominating set. Hence, for any minimal dominating set
  $D$ including $S$ and not intersect with $X$, the set of literals included in $D$ corresponds to a satisfiable assignment.  Therefore, the answer of the \extensionproblem is yes if and only if $\varphi$ has a satisfiable assignment.
\end{proof}

\begin{figure}
  \begin{center}
    \includegraphics[width=\textwidth]{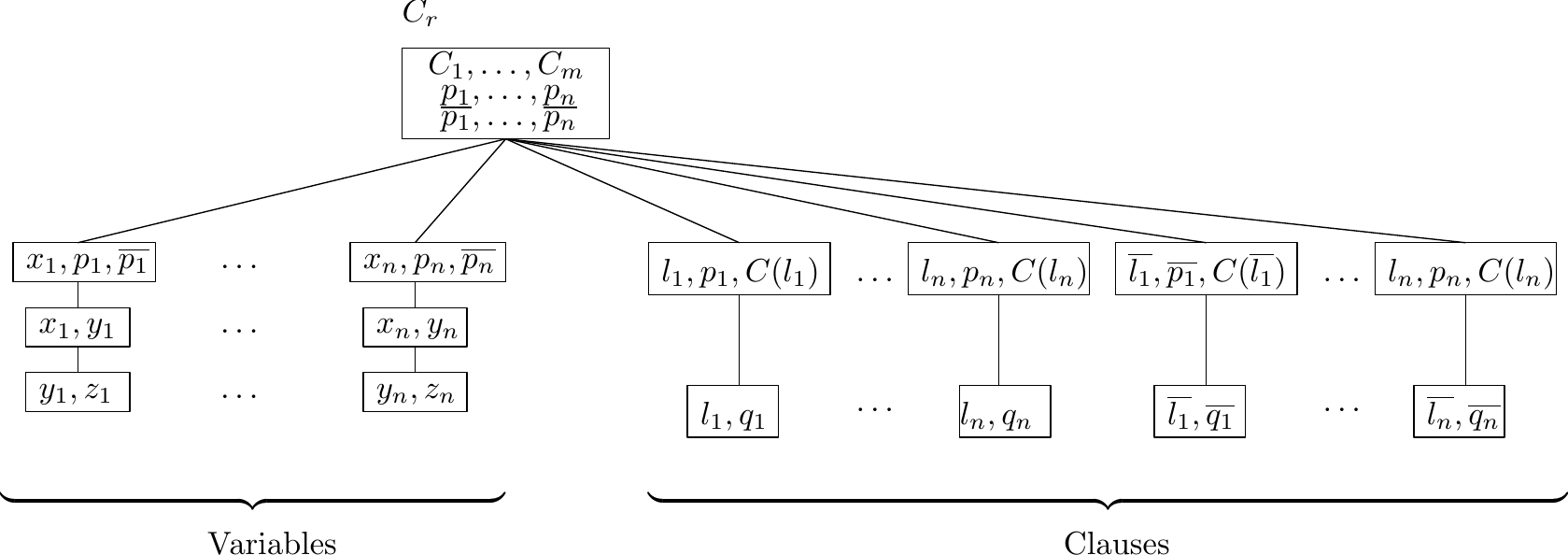}
    \end{center}
  \caption{An illustration of the construction of Theorem \ref{thm:np-chordal}.}
   \label{fig:np-chordal}
\end{figure}

To overcome these difficulties, we will follow another approach.  As we can see in the proof of the \np\!\!-completeness, when the root clique has both un-dominated vertices and private neighbors of several
vertices of $S$, the \extensionproblem turns to be difficult.  In the following, we will introduce a new strategy for the enumeration, that repeatedly enumerates antichains in levelwise manner. Indeed
for any minimal dominating set $D$ of a chordal graph $G$ the set $A(D)$ is an antichain that moreover dominates $Up(A(D))$.  Our strategy consists in enumerating such antichains and for each such
antichain $A$ enumerates the minimal dominating sets $D$ such that $A(D)=A$.  Let's be more precise in the next sections. 

\section{$(K_1,K_2)$-Extensions} \label{sec:kextension}

Along this section we consider a fixed chordal graph $G$ and clique tree $\cT$ of $G$ with root $C_r$ so that we do not need to recall them in the statements. 

Let $K_1, K_2\subseteq C_r$ be given disjoint sets that are decided to be included in the solution. In our setting $K_2$ will be the set of vertices that have already been assigned private neighbors
so that we do not need to search one for them.  Without confusion we denote $K_1\cup K_2$ by $K$.  A \emph{$(K_1,K_2)$-extension} of a partial antichain $A$ is a vertex set $D$ such that $(A\cup
K)\subseteq D$ and $D\setminus (A\cup K) \subseteq \bigcup\limits_{C\in \cL(A\cup K)} V(C)$.  Observe that if $D$ is a $(K_1,K_2)$-extension of $A$, then $A$ is a prefix of $A(D)$.  When the partial
antichain is not specified, $(K_1,K_2)$-extension is that for the empty partial antichain.  A $(K_1,K_2)$-extension $D$ is \emph{feasible} if it is a dominating set and $P(D,x)\ne \emptyset$ for all
$x\in D\setminus K_2$. A partial antichain $A$ is \emph{$(K_1,K_2)$-extendable} if it has a feasible $(K_1,K_2)$-extension.

For $C\in \cC_G$ and $x\in C$, let $\cF(C,x):=\{C'\preceq_\cT C$ and $C'\in \cL'(x)\}$, and let $D_C(x)$ denote a vertex set composed of
\begin{enumerate}
\item $Z\subseteq V(C)\cap \left(\bigcup\limits_{C'\in \cF(C,x)} C'\right)$ such that $|Z\cap C'| = |Z\cap f(C')| = 1$ for all $C'\in \cF(C,x)$,
\item a greedily obtained minimal dominating set of $G[(V(C)\setminus N[x])\setminus N[Z]]$.
\end{enumerate}

If $x\notin C$, then we let $D_C(x)$ be a greedily obtained minimal dominating set of $G[V(C)]$.

\begin{property}[Irredundancy of $D_C(x)$]\label{property:irredundant} Let $C\in \cC_G$ and let $x\in V_G$. Then $D_C(x)$ is an irredundant set in $G[V(C)]$. \end{property}

\begin{proof} Since each minimal dominating set is also an irredundant set, we can assume that $x\in C$. By definition of $Z$ we have that $\{x\}\times Z\cap E_G = \emptyset$. Moreover, by Property
  \ref{property:cliquetree}(2) no two vertices of $Z$ are adjacent. Since by construction of $D_C(x)\setminus Z$ no vertex in $D_C(x)\setminus Z$ is adjacent to a vertex of $Z$, we can conclude that for each
  $z\in Z$ we have $z\in P(D_C(x),z)$. Moreover, since $(D_C(x)\setminus Z) \cap N[Z]=\emptyset$ and $D_C(x)\setminus Z$ is a minimal dominating set of $G[(V(C)\setminus N[x])\setminus N[Z]]$, we can
  conclude that $P(D_C(x),y)\ne \emptyset$ for all $y\in (D_C(x)\setminus N[x])\setminus N[Z]$.
\end{proof}

\begin{property}[{Domination of $V(C)\setminus N[x]$}]\label{property:domination} Let $C\in \cC_G$ and let $x\in V_G$. Every vertex in $V(C)\setminus N[x]$ is dominated by $D_C(x)$.  \end{property}

\begin{proof} If $x\notin C$, then $D_C(x)$ is a minimal dominating set of $G[V(C)]$ and then we are done. So, assume that $x\in C$ and let $y\in V(C)\setminus N[x]$. Then $C(y)$ is necessarily a
  descendant of a clique $C'\in \cL'(x)$ and such that $C'\preceq_{\cT} C$. So, either $y\in N[Z]$ or $y\notin N[Z]$. In both cases, it is dominated by $D_C(x)$.
\end{proof}

Given disjoint sets $K_1, K_2\subseteq C_r$, $D\subseteq V_G\setminus K$ and $x\in D\cup K_1$, a vertex $y\in P(D\cup K,x)$ is said \emph{safe} if either $x=y$, or the following two
conditions are satisfied
\begin{enumerate}
\item[(S1)] $N(y)\cap V(C) \subseteq N[D_C(y)]$ for all $C\in \cL'(D\cup K)$ with $y\in C$ and 
\item[(S2)] for each $z\in N[y]\cap Uncov(D\cup K)$, there is a clique $C\in \cL'(D\cup K)$ such that $z\in N[D_C(y)]$. 
\end{enumerate}

A vertex $x\in D$ is said \emph{safe} if one of its private neighbors is safe. 

\begin{property}[{Domination of $V(C)\setminus \{y\}$ for safe $y$}]\label{property:dom} Let $x\in D\cup K_1$ and let $y\in P(D\cup K,x)$ be a safe for $x$. Then $V(C)\setminus \{y\}\subseteq
  N[D_C(y)]$ for all $C\in \cL'(D\cup K)$ with $y\in C$. 
\end{property}

\begin{proof} By Property \ref{property:domination} $V(C)\setminus N[y]$ is dominated by $D_C(y)$. By definition of safety $N(y)$ is dominated by $D_C(y)$. Therefore $V(C)\setminus \{y\}$ is
  dominated by $D_C(y)$ for all $C\in \cL'(D\cup K)$ with $y\in C$. 
\end{proof}

We will now prove some technical lemmas that will be used to prove the correctness of the algorithm. 

\begin{lem}[Extension Safe]\label{lem:extension-safe} Let $A$ be a partial antichain and let $x\in A\cup K_1$. For $y\in P(A\cup K,x)$ that is non-safe, no $(K_1,K_2)$-extension $D$ of $A$ that is a dominating set
  satisfies that $y\in P(D,x)$.
\end{lem}

\begin{proof} Since $y$ is not safe, we have $x\ne y$, and therefore $y$ violates one of the two conditions (S1) or (S2) to be safe.  Suppose that (S1) is not satisfied, \ie there is a clique $C\in
  {\cL}'(A\cup K), y\in C$ such that there is a vertex $z$ in $(N(y)\cap V(C))\setminus N[D_C(y)]$.  
  Thus, any $(K_1,K_2)$-extension $D$ of $A$ that is a dominating set includes some vertices in $N[y]$ other than $x$, thus $y$ is not a private neighbor of $x$.

  Suppose now that (S2) is not satisfied, \ie there is a vertex $z\in N[y]\cap Uncov(A\cup K)$ such that no clique $C\in {\cL}'(A\cup K)$ satisfies $z\in N[D_C(y)]$.  It implies from the
  definition of $D_C(y)$ that no vertex in $V(C)\setminus N[y]$ is adjacent to $z$ in all cliques $C\in {\cL}'(A\cup K)$. Thus, as in the previous case, in any $(K_1,K_2)$-extension $D$ of $A$, $y$ is not a
  private neighbor of $x$ unless $D$ is not a dominating set.
\end{proof}

\begin{lem}[Lower Private Neighbor] \label{lem:lowerpn} Let $A$ be a partial antichain and let $x\in A\cup K_1$ be safe. Then there is $y\in P(A\cup K,x)$ that is safe and such that $y\in V(C(x))$.
\end{lem}

\begin{proof} The statement holds if $x\in P(A\cup K,x)$.  If not, $C(x)$ includes another vertex in $A\cup K$, and it is adjacent to any vertex in $N[x]\setminus V(C(x))$ by Property
  \ref{property:cliquetree}.  Thus all its safe private neighbors are always in $V(C(x))$.
\end{proof}

\begin{lem}[Extendability of Partial Antichain]\label{lem:extend-pa}
 A partial antichain $A$ is $(K_1,K_2)$-extendable if and only if the following
  two conditions are satisfied
\begin{enumerate}
\item any vertex in $Uncov(A\cup K)$ is included in a clique of ${\cL}'(A\cup K)$, 
\item all vertices in $A\cup K_1$ are safe.
\end{enumerate}
\end{lem}

\begin{proof} Let $A$ be a $(K_1,K_2)$-extendable partial antichain. If (1) is not satisfied, there is a vertex $z\in Uncov(A\cup K)$ that is not included in
  any clique of $\cL'(A\cup K)$, and by definition of $(K_1,K_2)$-extension no $(K_1,K_2)$-extension of $A$ can dominate it.  So (1) is always satisfied. Now, if (2) is not satisfied, there is a non-safe vertex $x$ in
  $A\cup K_1$, thus all $y\in P(A\cup K,x)$ are non-safe.  By Lemma \ref{lem:extension-safe} it follows that $P(D,x)=\emptyset$ for each $(K_1,K_2)$-extension $D$ of $A$ that is a dominating set, and then (2)
  is always satisfied.

  Suppose now that the two conditions hold.  For each $x\in A\cup K_1$ let us choose one safe private neighbor and let us denote the set of all these safe private neighbors by $S$.  We consider a
  $(K_1,K_2)$-extension $D$ generated from $A\cup K$ as follows. First of all notice that from the definition of private neighbor and safety for each $C\in \cL'(A\cup K)$, $|C\cap S|\leq 1$. So, let
  $\cL_1:=\{C\in \cL'(A\cup K)\mid |C\cap S| = 1\}$ and $\cL_0:=\{C\in \cL'(A\cup K)\mid |C\cap S| =0\}$. It is clear that $\{\cL_0,\cL_1\}$ is a bipartition of $\cL'(A\cup K)$. Let $z\in S$. Now let
  \begin{align*} D & := (A\cup K)\cup \left(\bigcup\limits_{C\in \cL_1, C\cap S=\{y\}} D_C(y)\right) \cup \left(\bigcup\limits_{C\in \cL_0} D_C(z)\right). \end{align*} 

  $D$ is clearly a $(K_1,K_2)$-extension of $A$. By definition of $D_C(y)$ for each vertex $x\in D\setminus (A\cup K)$ we have that $P(D,x)\ne \emptyset$. It is moreover easy to check that for each
  $x\in A\cup K_1$, we have that $S\cap P(A\cup K,x) \in P(D,x)$.  Thus, from Property \ref{property:cliquetree}, $P(D,x)\ne \emptyset$ for all $x\in D\setminus K_2$. Each vertex in $N[A\cup K]$ is
  dominated.  Moreover, since for each $C\in \cL_0$ we have $z\notin C$, by definition of $D_C(z)$ we have $V(C)$ is also dominated.  Now, let $C\in \cL_1$ and let $C\cap S=\{y\}$. We know from
  Property \ref{property:domination} that $V(C)\setminus N[y]$ is dominated by $D_C(y)$ and $y$ is dominated by $A\cup K$ since $y$ is safe for some vertex in $A\cup K_1$. So, it remains to show
  that $N(y)\cap V(C)$ is dominated. By the definition of safety we know that the two conditions (S1) and (S2) are satisfied, \ie $N(y)\cap V(C)$ is dominated.
\end{proof}

As a corollary  we have the following.

\begin{lemma}\label{lem:check-ext-poly} For any partial antichain $A$ one can check in polynomial time whether $A$ is $(K_1,K_2)$-extendable.
\end{lemma}

\begin{proof} By Lemma \ref{lem:extend-pa} it is enough to check if~ (1) all vertices in $A\cup K_1$ are safe and~ (2) each vertex in $Uncov(A\cup K)$ is included in a clique in $\cL'(A\cup K)$.
  Since (2) can be easily checked in polynomial time from $G$ and a clique tree of $G$, it remains to show that (1) can be checked in polynomial time. A vertex $x\in A\cup K_1$ is safe if either $x\in
  P(A\cup K_1,x)$ or there exists a safe $y\in V(C(x))\cap P(A\cup K_1,x)$ by Lemma \ref{lem:lowerpn}. But by the definition of safety for each $y\in V(C(x))\cap P(A\cup K_1,x)$ the conditions (S1)
  and (S2) are of course checkable in polynomial time from $G$ and a clique tree of $G$.
\end{proof}

\section{The Algorithm} \label{sec:algo}

As in the previous section let us assume we are given a chordal graph $G$ and a clique tree $\cT$ of $G$ rooted at $C_r$. Remind that for a subset $S$ of $V_G$ the \emph{top antichain} of $S$ denoted
by $A(S)$ is the set of vertices of $S$ included in the upmost cliques in $\cC(S)$ that are not descendants of any other in $\cC(S)$, i.e., $A(S) := \{ v\in S \mid C(v) \mbox{ is in}
\max\limits_{\preceq_\cT}\{ \cC(S)\} \}$.  We observe that for any minimal dominating set $D$ of $G$, its top antichain is an $(\emptyset,\emptyset)$-extendable antichain.  Moreover, $D\setminus
A(D)$ is composed of vertices below $A(D)$, i.e., any vertex in $D\setminus A(D)$ is included in $V(C)\setminus C$ for some $C\in \cC(D)$.  Using this, we partition the minimal dominating sets
according to their top antichains.  Since these top antichains are $(\emptyset,\emptyset)$-extendable, we enumerate all $(\emptyset,\emptyset)$-extendable antichains, and for each
$(\emptyset,\emptyset)$-extendable antichain $A$, enumerate all minimal dominating sets whose top antichain is $A$.  As by definition of $(K_1,K_2)$-extendable for some disjoint $K_1,K_2\subseteq
C_r$, for each $(\emptyset,\emptyset)$-antichain $A$ there is at least one minimal dominating set whose top antichain is $A$.  Therefore, each output $(\emptyset,\emptyset)$-antichain will give rise
to a solution.  This is one of the key to polynomial delay.


Now for a minimal dominating set $D$ and a clique $C\in \cC(A(D))$, each vertex $x$ in $D\cap (V(C)\cup C)$ cannot have a private neighbor in another $G[V(C')\cup C']$ for some other $C'\in
\cC(A(D))$. Therefore, we can treat each $G[V(C)\cup C]$ independently. However, for each $C\in \cC(A(D))$ the set $D\cap (V(C)\cup C)$ is not necessarily a minimal dominating set of $G[V(C)\cup C]$
since $D\cap C$ may be equal to a singleton $\{x\}$ with $x$ having a private neighbor in $Up(A(D))$. In such cases we are looking in $G[V(C)\cup C]$ a dominating set $D'$ of $G[V(C)\cup C]$
containing $x$ where $x$ does not necessarily have a private neighbor, but all the other vertices in $D'$ do, \ie $D'$ is a feasible $(\{x\},\emptyset)$-extension in $G[V(C)\cup C]$ with clique tree
$\cT(C)$.  This situation is what exactly motivated the notion of $(K_1,K_2)$-extensions.  

Assume now we are given a pair $(K_1,K_2)$ of disjoint sets in $C_r$ and a $(K_1,K_2)$-extendable antichain $A$. Now contrary to $(\emptyset,\emptyset)$-antichains we can have a vertex $x$ in
$K:=K_1\cup K_2$ that belongs to several cliques in $A$. So we cannot independently make recursive calls in $G[V(C)\cup C]$ for each $C\in \cC(A)$. But, for each feasible $(K_1,K_2)$-extension of $A$ and
each $C\in \cC(A)$ the set $D\cap (V(C)\cup C)$ is a feasible $(K_C^1, K_C^2)$-extension of $G[V(C)\cup C]$ for some disjoint $K_C^1$ and $K_C^2$ in $(A\cup K)\cap C$. Now the whole task is to define for each
$C\in \cC(A)$ the sets $K_C^1$ and $K_C^2$ in $(A\cup K)\cap C$ in such a way that by combining all these feasible $(K_C^1,K_C^2)$-extensions we obtain a feasible $(K_1,K_2)$-extension of $A$, and also any
feasible $(K_1,K_2)$-extension can be obtained in that way. Actually, the way of setting $K_C^1$ and $K_C^2$ is the key, and is described below.  After defining $K_C^1$ and $K_C^2$ we will be able to
enumerate all the feasible $(K_C^1,K_C^2)$-extensions in $G[V(C)\cup C]$ in the same way.

In summary, our enumeration strategy is composed of nested enumerations: enumeration of $(K_1,K_2)$-extendable antichains, for each $(K_1,K_2)$-extendable antichain $A$ and each $C\in \cC(A)$ define
$K_C^1$ and $K_C^2$ and enumerate all the feasible $(K_C^1,K_C^2)$-extensions, and finally the combinations of all these $(K_C^1,K_C^2)$-extensions.  Since any minimal dominating set is a feasible extension
of some $(\emptyset,\emptyset)$-extendable antichain, the completeness of the enumeration is trivial. The rest of the section is as follows. We first show how to enumerate $(K_1,K_2)$-extendable
antichains for some fixed $(K_1,K_2)$. Then we show, given a $(K_1,K_2)$-extendable antichain $A$, how to define $K_C^1$ and $K_C^2$ for each $C\in \cC(A)$ and how to combine all the feasible
$(K_C^1,K_C^2)$-extensions in order to obtain all feasible $(K_1,K_2)$-extensions of $A$.  Before assuming that we can perform both tasks with polynomial delay and use only polynomial space let us
show that we can enumerate with polynomial delay and use polynomial space all the feasible $(K_1,K_2)$-extensions. 

\subsection{Enumeration of $(K_1,K_2)$-Extensions}

This subsection deals with the algorithm for enumerating all the feasible $(K_1,K_2)$-extensions, including the case of the root of the recursion.  As we explained, the algorithm is composed of
$(K_1,K_2)$-extendable antichain enumeration and of the enumeration of combinations of the feasible $(K_C^1,K_C^2)$-extensions for appropriate $(K_C^1,K_C^2)$.  It can be described as follows.

\begin{framed}
\begin{small}
\begin{tabbing}
{\bf Algorithm} {\sf EnumKExtension}$(G, \cT, K_1, K_2)$\\
\ \ \ \ \ \ \ $G$:graph, $\cT$:clique tree, $K_1,K_2$: disjoint subsets of $C_r$ root of $\cT$\\
1. {\bf for} each antichain $A$ output by {\sf EnumAntichain}$(G, \cT, K_1, K_2, \emptyset)$ {\bf do}\\
2. \ \ \ {\bf output} each solution of {\sf EnumCombination}$(G, \cT, K_1, K_2, A, A\cup K)$\\
3. {\bf end for}
\end{tabbing}
\end{small}
\end{framed}

Assume that {\sf EnumAntichain}$(G, \cT, K_1, K_2, \emptyset)$ enumerates all $(K_1,K_2)$-extendable antichains (Lemma \ref{lem:correctness-enum-antichain}) and {\sf EnumCombination}$(G, \cT, K_1,
K_2, A, A\cup K)$ enumerates all feasible $(K_1,K_2)$-extensions of $A$ (Lemma \ref{lem:correctness-enum-combination}), both with polynomial delay and use polynomial space. Then we have the following.

\begin{theorem}\label{thm:t1} The call {\sf EnumKExtension} $(G,\cT,K_1,K_2)$ enumerates all feasible $(K_1,K_2)$-extensions in polynomial delay and uses polynomial space.
\end{theorem}

\begin{proof} By definition for every feasible $(K_1,K_2)$-extension $D$ the top antichain $A(D\setminus K)$ is a $(K_1,K_2)$-extendable antichain. So by Lemmas \ref{lem:correctness-enum-antichain}
  and \ref{lem:correctness-enum-combination} below every feasible $(K_1,K_2)$-extension is output. Therefore, {\sf EnumKExtension} $(G,\cT,K_1,K_2)$ enumerates all feasible $(K_1,K_2)$-extensions.
  From the definition of $(K_1,K_2)$-extendable antichains every call in Step 1 outputs at least one feasible $(K_1,K_2)$-extension.  Now since {\sf EnumAntichain}$(G, \cT, K_1, K_2, \emptyset)$ and\\
  {\sf EnumCombination}$(G, \cT, K_1,K_2, A, A\cup K)$ runs with polynomial delay and use both polynomial space we can conclude that {\sf EnumKExtension} $(G,\cT,K_1,K_2)$ runs with polynomial delay
  and use polynomial space.
\end{proof}

\subsection{Enumeration of Antichains}

Our strategy is to enumerate all $(K_1,K_2)$-extendable partial antichains
 by an ordinary backtracking algorithm, that repeatedly appends a
  vertex to the current solution that is larger than its tail.
In this algorithm, any $(K_1,K_2)$-extendable partial antichain $A$ is
 obtained from $A\setminus tail(A)$.
Since $A\setminus tail(A)$ is a prefix of $A$, any $(K_1,K_2)$-extendable 
 partial antichain is generated from another $(K_1,K_2)$-extendable partial 
 antichain.
This implies that the set of $(K_1,K_2)$-extendable partial antichains satisfies
 a kind of monotone property, and thus we can enumerate all
 $(K_1,K_2)$-extendable partial antichains with passing through only
 $(K_1,K_2)$-extendable partial antichains.
The algorithm is described as follows.

\begin{framed}
\begin{small}
\begin{tabbing}
{\bf Algorithm} {\sf EnumAntichain}$(G, \cT, K_1, K_2, A)$\\
\ \ \ \ \ \ \ $G$:graph, $\cT$:clique tree, $K_1,K_2$: disjoint subsets of $C_r$ root of $\cT$\\ \ \ \ \ \ \ \ \ $A$:$(K_1,K_2)$-extendable partial antichain\\
1. {\bf if} $A$ is an antichain {\bf then output} $A$;\\
2. {\bf for} each vertex $z>tail(A)$ {\bf do}\\
3. \ \ {\bf if} $A\cup \{ z\}$ is a $(K_1,K_2)$-extendable partial antichain {\bf then}\\
\ \ \ \ \ {\bf call} {\sf EnumMinAntichain}$(G, \cT, K_1, K_2, A\cup \{z\})$\\
4. {\bf end for}
\end{tabbing}
\end{small}
\end{framed}

\begin{lem}\label{lem:correctness-enum-antichain}
The call {\sf EnumAntichain}$(G, \cT, K_1, K_2, \emptyset)$ enumerates all
 $(K_1,K_2)$-extendable antichains in polynomial delay with polynomial space.
\end{lem}

\begin{proof} We observe that for any $(K_1,K_2)$-extendable partial antichain $A$, $A\setminus tail(A)$ is a $(K_1,K_2)$-extendable partial antichain.  Thus, one can easily prove by induction that
  the iteration inputting $A$ is recursively called only by the iteration inputting $A\setminus tail(A)$. Therefore, all $(K_1,K_2)$-extendable partial antichains are generated by this algorithm
  without repetition.  For a $(K_1,K_2)$-extendable partial antichain $A$, there is at least one feasible $(K_1,K_2)$-extension $D$.  By the definition of a feasible $(K_1,K_2)$-extension,
  $A(D\setminus K)$ is a $(K_1,K_2)$-extendable antichain with $A$ as a prefix.  This implies that at least one descendant of any iteration outputs an antichain, and every leaf of the recursion tree
  outputs an antichain.  Then, the delay is bounded by the maximum computation time of an iteration multiplied by the depth of the recursion.  The depth is at most $|V_G|$, thus the algorithm is
  polynomial delay since the loop at Step 2 runs at most $n$ times and the $(K_1,K_2)$-extendability check can be done in polynomial time by Lemma \ref{lem:check-ext-poly}. Since the depth is bounded
  by $|V_G|$, the algorithm uses obviously a polynomial space.
\end{proof}

\subsection{Enumeration of Combinations}

We now show, given a $(K_1,K_2)$-extendable antichain $A$, how to enumerate with polynomial delay and use only polynomial space all feasible $(K_1,K_2)$-extensions of $A$ by computing for each $C\in
\cC(A)$ all the $(K_C^1,K_C^2)$-extensions of $G[V(C)\cup C]$ for appropriate $K_C^1$ and $K_C^2$ and combine all of them. Note that the set $A$ is the top antichain of any feasible
$(K_1,K_2)$-extension if and only if if the $(K_1,K_2)$-extension is that of $A$.  For pruning redundant partial combinations, we introduce the notion of a partial $(K_1,K_2)$-extension.  A vertex set
$D\supseteq A\cup K$ is called a \emph{partial $(K_1,K_2)$-extension} of $A$ if there is a feasible $(K_1,K_2)$-extension $D'$ of $A$ such that $D\setminus (A\cup K)$ is a prefix of
$D'\setminus (A\cup K)$, and all the vertices in $V(C(x))$ for $x\in A$ is dominated by $D$ if $x$ is smaller than $tail(D\setminus (A\cup K))$. Our strategy is to enumerate all partial
$(K_1,K_2)$-extensions of $A$, similar to the antichain enumeration.  For a partial $(K_1,K_2)$-extension $D$ of $A$, let $C^*(D)$ be the largest clique $C$ in $\cC(A)$ such that
$(D\setminus (A\cup K))\cap V(C)\ne \emptyset$, and $C_*(D)$ be the smallest clique $C$ in $\cC(A)$ such that a vertex in $V(C)$ is not dominated by $D$.  Informally $C^*(D)$ is the last clique $C\in
\cC(A)$ such that $V(C)$ is dominated by $D$, and $C_*(D)$ the first clique in $\cC(A)$ such that $V(C)$ is not dominated by $D$. To enumerate all partial $(K_1,K_2)$-extensions of $A$
and in fine all $(K_1,K_2)$-extensions of $A$, we start from $D = A\cup K$ and repeatedly add a $(K_{C_*(D)}^1,K_{C_*(D)}^2)$-extension of $G[V(C_*(D))\cup C_*(D)]$ to $D$ for appropriate
$(K_{C_*(D)}^1,K_{C_*(D)}^2)$, while keeping the extendability.  To characterize the possible $(K_{C_*(D)}^1,K_{C_*(D)}^2)$ we state the following lemma.  Let $Q_D(C')$ be the vertices $x$ in
$K\cup A$ that have no safe private neighbor in $V(C)\cup C, C>C'$, and none of its private neighbor in $P(K\cup A\cup D, x)$ is included in $Up(A)\setminus C'$ or in $V(C), C<C'$. In other words
$Q_D(C')$ is the set of vertices in $K\cup A$ that we must give a private neighbor in $V(C')\cup C'$ for any $(K_1,K_2)$-extension of $A$ containing $D$. 

\begin{lemma}\label{lem:comb-parent} For a non-empty partial $(K_1,K_2)$-extension $D$, $D\cap(V(C^*(D))\cup C^*(D))$ is a feasible $(K_1',K_2')$-extension in $G[V(C^*(D))\cup C^*(D)]$ where $K_1' =
  Q_D(C^*(D))$ and $K_2' = ((A\cup K)\cap C^*(D))\setminus K_1'$.
\end{lemma}

\begin{proof} By definitions of partial $(K_1,K_2)$-extension and of $C^*$, $D\cap (V(C^*(D)\cup C^*(D))$ dominates $V(C^*(D))$.  Moreover, every vertex $x$ in $Q_D(C^*(D))$ has a private neighbor only in
  $V(C^*(D))\cup C^*(D)$, and moreover  $x\in C^*(D)$.  Thus, the statement holds.
\end{proof}

\begin{lem}{\label{lem:comb-desc}} Let $D$ be a partial $(K_1,K_2)$-extension of $A$ and suppose that $C_*(D)$ exists.  For any feasible $(K_1',K_2')$-extension $D'$ in $G[V(C_*(D))\cup C_*(D)]$
  where $K_1' = Q_D(C_*(D))$, $K_2' = ((A\cup K)\cap C_*(D))\setminus K_1'$, $D\cup D'$ is a partial $(K_1,K_2)$-extension of $A$.
\end{lem}

\begin{proof} As in the proof of Lemma \ref{lem:extend-pa}, we choose one private neighbor for vertices in $A\cup K$ that have safe private neighbors in $V(C), C>C_*(D)$ and let $S$ be the set of
  these selected vertices.  Then we let $\cL_1:=\{C\in \cL'(A\cup K)\mid C>C_*(D), |C\cap S| = 1\}$ and $\cL_0:=\{C\in \cL'(A\cup K)\mid C>C_*(D), |C\cap S| =0\}$.  Let $z\in S$. Now let
  
  \begin{align*} D^* & := (A\cup K \cup D\cup D')\cup \left(\bigcup\limits_{C\in \cL_1, C\cap S=\{y\}} D_C(y)\right) \cup \left(\bigcup\limits_{C\in \cL_0} D_C(z)\right). \end{align*} 
 
  According to the proof of Lemma \ref{lem:extend-pa}, $D^*$ is a feasible $(K_1,K_2)$-extension of $A$.
\end{proof}

We can now describe the algorithm.

\begin{framed}
\begin{small}
\begin{tabbing}
{\bf Algorithm} {\sf EnumCombination}$(G, \cT, K_1, K_2, A, D)$\\
\ \ \ \ \ \ \ $G$:graph, $\cT$:clique tree, $K_1,K_2$: disjoint subsets of $C_r$ root of $\cT$, $A$:$(K_1,K_2)$-extendable antichain\\
\ \ \ \ \ \ \ $D$: a partial $(K_1,K_2)$-extension of $A$\\
1. {\bf if } $C_*(D)$ does not exist {\bf then output} $D$; {\bf return}\\ 
2. $K_1' = Q_D(C_*(D))$, $K_2':= ((A\cup K)\cap C_*(D))\setminus K_1'$\\
3. {\bf for} each $D'$ output by {\sf EnumKExtension}$(G[V(C_*(D))\cup C_*(D)], \cT(C^*(D)), K_1', K_2')$\\
4. \ \ \ {\bf call} {\sf EnumCombination}$(G, \cT, K_1, K_2, A, D\cup D')$\\
5. {\bf end for}
\end{tabbing}
\end{small}
\end{framed}

\begin{lem}\label{lem:correctness-enum-combination} The call {\sf EnumCombination}$(G, \cT, K_1, K_2, A, A\cup K)$ enumerates
  all feasible $(K_1,K_2)$-extensions whose top antichain is $A$ in polynomial delay and uses polynomial space.
\end{lem}


\begin{proof} From Lemma \ref{lem:comb-parent}, the iteration inputting a partial $(K_1,K_2)$-extension $D$ of $A$ is generated only from the iteration inputting $D\setminus (V(C^*(D)\setminus
  C^*(D)))$.  This assures that the algorithm enumerates all partial $(K_1,K_2)$-extensions of $A$ without duplication.  From Lemma \ref{lem:comb-desc}, there is at least one feasible
  $(K_1,K_2)$-extension $D'$ of $A$ including the partial $(K_1,K_2)$-extension $D$ of $A$ that is the input of the iteration.  Thus, all the leaf iterations of the recursion of this algorithm always
  outputs a feasible $(K_1,K_2)$-extension of $A$.  Now the delay is bounded by the maximum computation time of an iteration multiplied by the depth of the recursion.  The depth is at most $|V_G|$,
  thus the algorithm is polynomial delay since {\sf EnumKExtension} runs with polynomial delay. Since the depth is at most $|V_G|$, the algorithm is obviously polynomial space.
\end{proof}

We are now ready to summarize the proof of our main theorem. 

\begin{proof}[Proof of Theorem \ref{thm:main}] By definition every minimal dominating set of $G$ is a feasible $(\emptyset,\emptyset)$-extension. Therefore,
the call {\sf EnumKExtension} $(G,\cT,\emptyset,\emptyset)$ enumerates all minimal dominating sets  in polynomial delay and uses polynomial space by Theorem \ref{thm:t1}.
\end{proof}

\section{Conclusion} 

We have proved that one can list all the minimal dominating sets of a chordal graph with polynomial delay and polynomial space. We know from \cite{GolovachHKKV14a} that there exists an output-polynomial
algorithm for the listing of minimal dominating sets of any chordal bipartite graph.   It is known that chordal bipartite graphs admit a tree-structure similar to the clique tree of chordal
graphs. Can we adapt our technique to obtain a polynomial delay and polynomial space algorithm for enumerating the minimal dominating sets of any chordal bipartite graph? 

Besides the fact that knowing whether there exists an output-polynomial algorithm for listing the set of minimal dominating sets in a graph is still open, there are some graph classes where a search
for an output-polynomial algorithm deserves to be explored and seems more tractable than the general problem: we can cite bipartite graphs and unit-disk graphs.  

By the reduction in \cite{KanteLMN12} we know that the enumeration of minimal dominating sets in co-bipartite graphs is as hard as the enumeration of minimal dominating sets in all graphs. Can we find
a parameter in co-bipartite graphs that whenever bounded by a function on $n$ (say $\log(n)$)  would summarize the tractability of the enumeration of minimal dominating sets in many graph classes? 

A related question that arises from the exact algorithm community is the existence of a tight bound of the number of minimal dominating sets in a graph. From \cite{FominGPS08} we know that the number of
minimal dominating sets in an $n$-vertex graph is bounded by $O(1.7159^n)$ and the best known lower bound is $15^{n/6}$. For several graph classes, including some subclasses of chordal graphs, tight
bounds were obtained  \cite{CouturierHHK13,GolovachHKKV14b}. Finding a tight upper bound for chordal graphs is still open. 

\bibliographystyle{plain}
\bibliography{bib}

\end{document}